\documentclass[aps,twocolumn,showpacs]{revtex4-1}
\usepackage{graphicx,bm,amsmath,amsthm, lmodern}

\newcommand{\aver}[2]{\ensuremath{\left\langle#1\right\rangle_{#2}}}
\newcommand{\kl}[2]{\ensuremath{{D}_{\rm KL}\left[#1\parallel#2\right]}}
\newcommand{\pd}[2]{\ensuremath{\frac{\partial #1}{\partial #2}}}
\newcommand{\tr}{{\ensuremath{\rm tr}\,}}
\newcommand{\capa}{\ensuremath{\mathcal{C}}}
\newcommand{\var}{{\ensuremath{\rm Var}}}
\newcommand{\corr}{{\ensuremath{\rm corr}}}
\newcommand{\T}{{\ensuremath{\sf T}}}
\newcommand{\sgn}{{\ensuremath{\rm sgn}}}
\newcommand{\normal}[2]{\ensuremath{{\mathcal N}\left(#1, #2\right)}}

\begin{document}
\title{Information capacity in the weak-signal approximation}

\author{Lubomir Kostal}
\email[]{kostal@biomed.cas.cz}

\affiliation{Institute of Physiology AS CR, v.v.i., 
Videnska 1083, 142 20 Prague 4, Czech Republic}

\date{August 10, 2010: Accepted for publication in Physical Review E}

\begin{abstract}
We derive an approximate expression for mutual information in a broad
class of discrete-time stationary channels with continuous input,
under the constraint of vanishing input amplitude or power. The
approximation describes the input by  its covariance matrix, while the
channel properties are described by the Fisher information matrix.
This separation of input and channel properties allows us to analyze
the optimality conditions in a convenient way.  We show that input
correlations in memoryless channels do not affect channel capacity
since their effect decreases fast with vanishing input amplitude or
power.  On the other hand, for channels with memory, properly matching
the input covariances to the dependence structure of the noise may
lead to almost noiseless information transfer, even for intermediate
values of the noise correlations.  Since many model systems described
in mathematical neuroscience and biophysics operate in the high noise
regime and weak-signal conditions, we believe, that the described
results are of potential interest also to researchers in these areas.
\end{abstract}

\pacs{89.90.+n, 89.70.Kn}
\maketitle

\section{Introduction}

Information theory is a mathematical framework that provides tools for
quantification of information content and information transfer in
systems defined by general probabilistic rules \citep{r:cover}.  The
theory has been applied successfully to a wide range of problems
\citep{r:verdurev}, including, e.g., classical and quantum computation
and communication \citep{r:bennett, r:cerf, r:gallager}, optical
communication \citep{r:mcdonnellflitney, r:mitra, r:turitsyn} or
quantification of different aspects of information processing in real
neurons and neuronal models \citep{r:atick, r:borstrev, r:laughlin,
r:mcdonnell08prl, r:koslanrosp08plos, r:rieke, r:stein-bj}. 

The measure of information transfer in information theory is
represented by a nonlinear functional of the probability measure over
the joint input-output space \citep{r:cover}.  The concavity of this
functional in the input probability measure has important implications
for numerical approaches to finding the information optimality
conditions \citep{r:cover, r:dauwels, r:ikeda, r:smith71}. On the
other hand, approximations or even closed-form solutions are quite
rare. The classical exact solution for the linear channel with
additive (possibly non-white) Gaussian noise \citep{r:cover, r:yeung}
and input power constraint has been applied in many different
situations. However, in many cases of interest the channel is
significantly nonlinear or non-Gaussian or there are different input
constraints \citep{r:mceliece} and one has to rely on numerical
solutions or approximations.

The  approximations allow us to investigate, although locally and
under perhaps restrictive scenario, the effect of individual
components in the system on the optimality conditions.  In particular,
if the noise in information transfer is substantially low and regular,
there exists a tight lower bound on the information optimality
conditions (denoted as \emph{low-noise} approximation in this paper)
which has been investigated in \citep{r:bernardo79, r:brunelnadal,
r:clarkebarron, r:mcdonnell08prl}.  In this paper we continue the
effort started in \citep{r:koslan10prer} and we describe essentially
the opposite situation: the \emph{high-noise} approximation. Such
approximation is of interest when the signal is very weak compared to
the noise in the information transfer, for example, as in the
classical stochastic resonance effect observed in electrosensory
neurons \citep{r:greenwoodprl, r:koslan10prer}.

\section{Measures of information}

Throughout this paper we assume the discrete-time setting
\citep{r:gallager}, we
denote the consequent channel outputs (responses)
as a vector of random variables (r.v.) $R=(\{R_i\}_{i=1}^{n})^\T$,
which may be discrete or continuous,  $i$ indexes the time and
$(\cdot)^{\T}$ denotes the transposition.  The response, $R_i=r_i$, 
results from the corresponding input $\Theta_i=\theta_i$, where the 
input is also described by a $n$-dimensional r.v.
$\bm\Theta$. The multidimensional description of the process of
information transfer between $\bm\Theta$ and $\mathbf R$
allows us to include the effect of memory, i.e.,
the  dependence on  current and also on past
inputs and responses.  We also assume that the input alphabet 
is continuous \citep{r:gallager}.
In the following we consider stationary channels 
fully described by the
conditional probability density function (p.d.f.)
$f(\mathbf r|\bm\theta)$, which generally
factorizes as \cite{r:ash}
\begin{equation}
f(\mathbf r|\bm\theta)
=
\prod_{i=1}^n
f_i (r_i| 
\theta_i, \theta_{i-1}, \dots, \theta_1,
r_{i-1}, \dots, r_1).
\label{eq:factor}
\end{equation}
We do not consider channel feedback, the
dependence of current input  on past responses \citep{r:cover}.

Mutual information (MI)
is the fundamental quantity measuring information transfer in
channels \citep{r:cover}.
MI $I(\bm\Theta;\mathbf R)$ gives the degree of
statistical dependence between inputs and responses,
defined as
\begin{equation}
I(\bm\Theta;\mathbf R)= 
\aver{\kl{f(\mathbf r|\bm\theta)}{p(\mathbf r)}}{\bm\theta},
\label{eq:mi}
\end{equation}
where 
\begin{equation}
p(\mathbf r)=\aver{f(\mathbf r|\bm \theta)}{\bm\theta}
\label{eq:rmarginal}
\end{equation}
is the marginal joint p.d.f. of responses, and the averaging is
with respect to the input p.d.f., $\pi(\bm\theta)$.
The Kullback-Leibler (KL) divergence is defined as
\begin{equation}
\kl{f(\mathbf r|\bm\theta)}{p(\mathbf r)}=
\aver{\ln \frac{f(\mathbf r|\bm\theta)}{p(\mathbf r)}}
{\mathbf r|\bm\theta},
\label{eq:kl}
\end{equation}
where the averaging is with respect to $f(\mathbf r|\bm\theta)$.
From Eq.~(\ref{eq:mi}) follows, that MI is a
property of the joint distribution of stimuli and
responses.
Of particular interest are the
\emph{optimality conditions}
for information transfer, that is the maximum
value of $I(\bm\Theta; \mathbf R)$
and the  corresponding
optimal $\pi(\bm\theta)$. In order to have a well-posed problem,
one is interested in the optimality conditions for $\bm\Theta$
satisfying certain additional constraints, e.g., average power or 
range of inputs \citep{r:cover, r:mceliece}.
The maximum value of MI per channel use, taken over all
possible stimuli distributions satisfying constraints $\mathcal G$,
is denoted as the information capacity, $\capa$, defined as
\citep{r:mceliece}
\begin{equation}
\capa= \lim_{n\rightarrow\infty} \frac{1}{n}
\left[\sup_{\pi(\bm\theta)\in\mathcal G}
 I(\bm\Theta;\mathbf R) \right].
\label{eq:capa}
\end{equation}
In this paper we interpret $\capa$  
as the upper bound on the rate
at which the information can be transmitted reliably \citep{r:cover},
without
considering the complexity of achieving such maximum rate in practical
terms. Specifically,
do not discuss  the properties of any particular coding and decoding
schemes \citep{r:gallager}. 

Whenever we are interested in reliability of
input-output transmission, we naturally interfere with the
domain of statistical estimation theory \citep{r:kay}. 
Fisher information (FI) matrix, defined as
\begin{equation}
\mathbf J(\bm\theta|\mathbf R)=
\aver{[\nabla \ln f(\mathbf r|\bm\theta)]
[\nabla \ln f(\mathbf r|\bm\theta)]^\T}{\mathbf r|\bm\theta},
\label{eq:fi}
\end{equation}
where
\begin{equation}
\nabla= \left(\pd{}{\theta_1}, \cdots, \pd{}{\theta_n} \right)^\T,
\end{equation}
imposes  limits on the precision of
$\bm\theta$ estimation from the responses by means of the Cramer-Rao
bound, which says 
that for the variance of any unbiased estimator of $\theta_i$ holds
$\var(\hat{\theta}_i)\geq [\mathbf J^{-1}(\bm\theta|\mathbf R)]_{ii}$
\citep{r:kay}.
Generally, FI requires that $f(\mathbf r|\bm\theta)$
is continuously differentiable in $\bm\theta$ \citep{r:kay}. In this
paper, we additionally assume that $f(\mathbf r|\bm\theta)$
is twice continuously differentiable in $\bm\theta$, 
so that the following  conditions hold
\begin{equation}
\int_{\mathbf R} \nabla
f(\mathbf r| \bm\theta)\,d\mathbf r=\mathbf 0,
\quad
\int_{\mathbf R} \nabla\nabla^{\T}
 f(\mathbf r| \bm\theta)\,d\mathbf r=\mathbf 0.
\label{eq:cdiff}
\end{equation}

There is a variety of relationships between FI, MI and KL divergence
established in the literature \citep{r:cover,r:kullback, r:salicru},
further motivated by the fields of information geometry
\citep{r:amari} or stochastic complexity \citep{r:rissanen}.  The
already mentioned \emph{low-noise} approximation to 
MI is constructed by employing the
Cramer-Rao bound \citep{r:bernardo79, r:brunelnadal, r:clarkebarron,
r:mcdonnell08prl}.  Although we demonstrate that the \emph{high-noise}
approximation also involves FI, we never employ the Cramer-Rao bound
and the appearance of FI is  due to certain asymptotic properties
of the KL distance \citep{r:kullback}.

\section{Information transfer by weak signals}

\subsection{Small input amplitude limit}

The channel properties are described by the conditional probability
density $f(\mathbf r| \bm\theta)$, which satisfies the regularity 
conditions~(\ref{eq:cdiff}).
The input, described by r.v. $\bm\Theta$, is restricted in
amplitude,
\begin{equation}
\bm\Theta\in [\bm\theta_0-\Delta\bm\theta,
\bm\theta_0+\Delta\bm\theta],
\label{eq:thetares}
\end{equation}
for chosen $\bm\theta_0$ and $\Delta\bm\theta$, or more precisely in
components:
for all $i$ holds $\Theta_i\in [\theta_0 -\Delta\theta, \theta_0
+\Delta\theta]$ and $\Delta\theta>0$. 
The situation for a memoryless channel
is illustrated in Fig.~\ref{fig:illu}.
The goal is to derive
an approximation to mutual information in the
limit $\|\Delta\bm\theta\|\rightarrow 0$.
We demonstrate in detail in Appendix~A, that the approximation 
(to second order in the input amplitude) can be written as
\begin{equation}
I(\bm\Theta; \mathbf R) 
\approx
\frac{1}{2} \tr\left[
\mathbf J(\bm\theta_0|\mathbf R)
\mathbf C_{\bm\Theta}
\right],
\label{eq:weakmi}
\end{equation}
where  $\mathbf J(\bm\theta_0|\mathbf R)$ is the FI matrix from
Eq.~(\ref{eq:fi}) evaluated at $\bm\theta=\bm\theta_0$,
$\mathbf C_{\bm\Theta}$ is the
covariance matrix of $\bm\Theta$ and $\tr(\cdot)$ is the matrix trace.
Eq.~(\ref{eq:weakmi}), derived also in \cite{r:koslan10prer},
holds for a broad class of channels with
memory, both biologically-inspired and artificial and represents
the main result. 
An important feature of Eq.~(\ref{eq:weakmi}) is, that 
the channel properties (described by the FI matrix) and the input
properties (described by its covariance matrix) are separated.
Therefore, the maximum value of MI  can be found by matching
the corresponding elements of $\mathbf J(\bm\theta_0|\mathbf R)$ and
$\mathbf C_{\bm\Theta}$. The elements of the
covariance matrix of $\bm\Theta$
can be written  as \citep{r:kendall3}
\begin{equation}
[\mathbf C_{\bm\Theta}]_{ik}= \sigma^2 \varrho_{ik},
\label{eq:cm}
\end{equation}
where $\sigma^2\equiv  \sqrt{\var(\Theta_i) \var(\Theta_k)}$
is constant for all $i,k$ due to
stationarity, and $\varrho_{ik}= \corr(\Theta_i, \Theta_k)$ is the
correlation coefficient. The maximum variance of the 
amplitude constrained input from Eq.~(\ref{eq:thetares})
is $\max \sigma^2= (\Delta\theta)^2$ and $-1< \varrho_{ik} < 1$,
thus $I(\bm\Theta; \mathbf R)$ in Eq.~(\ref{eq:weakmi}) is maximized
if
\begin{equation}
\varrho_{ik}\rightarrow \sgn [\mathbf J(\bm\theta_0|\mathbf R)]_{ik},
\label{eq:dcov}
\end{equation}
where $\sgn(\cdot)$ is the sign function.
Note, that the 
diagonal elements of the FI matrix are positive while the
off-diagonal elements can be negative.
It may happen, that
the matrix $\mathbf C_{\bm\Theta}$ formed by Eqns.~(\ref{eq:dcov})
and~(\ref{eq:cm})
is not positive-semidefinite\footnote{
Consider for example the matrix
$\left(
\begin{smallmatrix}
2 & -1 & -1 \\ -1 & 2 & -1 \\ -1 & -1 &2
\end{smallmatrix}\right)
$,
which is positive-semidefinite, while
$
\left(
\begin{smallmatrix}
1 & -1 & -1 \\ -1 & 1 & -1 \\ -1 & -1 &1
\end{smallmatrix}\right)
$
is not. If the desired covariance matrix cannot be formed
from Eq.~(\ref{eq:dcov}), then 
$\capa_{\rm high}$ is less than as given by Eq.~(\ref{eq:weakcapa}).
},
i.e., it cannot be a proper covariance 
matrix \citep{r:horn},
even though $\mathbf J(\bm\theta_0|\mathbf R)$  generally is
positive-semidefinite \citep{r:kay}.
However, in all problems we have calculated so far, proper 
input covariance matrix could be formed, given 
$\mathbf J(\bm\theta_0|\mathbf R)$, and then it holds from
Eqns.~(\ref{eq:capa}) and~(\ref{eq:weakmi})
\begin{equation}
\capa\approx \capa_{\rm high}=
\lim_{n\rightarrow\infty} \frac{(\Delta\theta)^2}{2n}
\sum_{i,k}\left|
[\mathbf J(\bm\theta_0|\mathbf R)]_{ik}
\right|,
\label{eq:weakcapa}
\end{equation}
where $\capa_{\rm high}$ denotes the \emph{high noise} approximation
to the true capacity $\capa$.

\begin{figure}
\begin{center}
\includegraphics[]{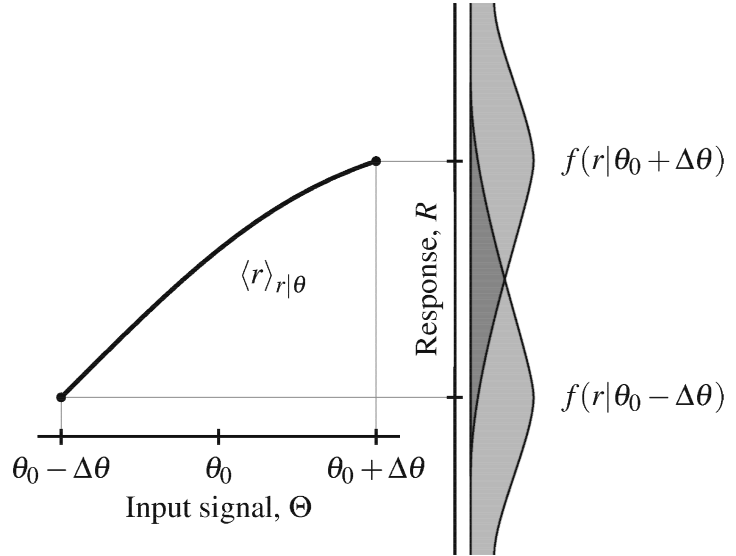}
\end{center}
\newpage
\caption{Information transmission with amplitude-constrained
inputs. The input signal, described by r.v. $\Theta$, is restricted 
to the interval $[\theta_0-\Delta\theta,
\theta_0+\Delta\theta]$. Due to presence of noise, the responses to each
particular $\theta$ vary randomly,  described by the conditional
probability density $f(r|\theta)$. While the memoryless
information channel is
fully described by $f(r|\theta)$, the amount of information
transferred depends on both $f(r|\theta)$ and the distribution of
$\Theta$. We examine the maximum information transfer
by inputs restricted to small amplitudes
 when there is a significant overlap of $f(r|\theta_0- \Delta\theta)$
and $f(r| \theta_0+\Delta\theta)$.
Heuristically, the problem can be also described as the information
transmission in a very noisy environment, or under very low
signal-to-noise ratio conditions.
}
\label{fig:illu}
\end{figure}

For stationary 
memoryless channels $f(\mathbf r|\bm\theta)$ factorizes
due to Eq.~(\ref{eq:factor}) as \citep[p.193]{r:cover}
\begin{equation}
f(\mathbf r|\bm\theta)= \prod_{i=1}^n f(r_i| \theta_i), 
\label{eq:memlfact}
\end{equation}
thus from Eq.~(\ref{eq:fi}) follows that
the FI matrix is diagonal, $J(\theta_0|R)\equiv
[\mathbf J(\bm\theta_0|\mathbf R)]_{ii}= 
\aver{[\partial_{\theta} \ln f(r|\theta)]^2}{r|\theta}$, and from
Eq.~(\ref{eq:weakcapa}) we have
\begin{equation}
\capa_{\rm high}= \frac{(\Delta\theta)^2}{2} J(\theta_0|R),
\label{eq:capamemless}
\end{equation}
a result obtained by different means in \citep{r:verdu-capcost}.
The optimal input p.d.f., $\pi^*(\theta)$, is the maximum variance
distribution over the given input range, 
\begin{equation}
\pi^*(\theta)= \frac{1}{2} \delta(\theta- \theta_0- \Delta\theta)
+ \frac{1}{2} \delta(\theta- \theta_0 +\Delta\theta),
\label{eq:bin}
\end{equation}
where $\delta(\cdot)$ is the Dirac's delta function. In other words,
the capacity is achieved by a binary input, and thus 
$\capa\leq 1$\,bit.

From Eq.~(\ref{eq:weakmi}) follows, that
non-diagonal elements of $\mathbf C_{\bm\Theta}$ do
not affect the information capacity of memoryless channels in the
vanishing input amplitude case. This result is counterintuitive,
because correlations generally decrease the input entropy
\citep{r:cover}. Therefore in the following we 
provide a proof which is independent of Eq.~(\ref{eq:weakmi}).
Let us  consider two consequent uses of a  stationary
memoryless channel, i.e., $\bm\Theta=\{\Theta_1, \Theta_2\}^\T$,
$\mathbf R=\{R_1, R_2\}^\T$. We assume, that the inputs
$\Theta_1$ and $\Theta_2$ are generally statistically dependent,
$(\Theta_1, \Theta_2)\sim \pi(\theta_1, \theta_2)$, and the joint
marginal distribution of responses is denoted as $p(\mathbf r)$,
see also Eq.~(\ref{eq:rmarginal}).
By employing the factorization~(\ref{eq:memlfact}) and
basic relations between entropy, 
$h(\mathbf R)= -\aver{\ln p(\mathbf r)}{\mathbf r}$,
and MI
\citep[p.21]{r:cover} we have
\begin{eqnarray}
I(\bm\Theta; \mathbf R) &=&
h(\mathbf R)- \aver{h(\mathbf R| \bm\theta)}{\bm\theta} 
=\nonumber\\
&=&h(R_1) +h(R_2) - I(R_1; R_2)- 
\nonumber\\
&&-\aver{h(R_1|\theta_1) + h(R_2|\theta_2)}{\bm\theta} 
=\nonumber\\
&=& I(\Theta_1; R_1) + I(\Theta_2; R_2) - I(R_1; R_2) 
=\nonumber\\
&=& 2I(\Theta_1; R_1) - I(R_1; R_2),
\label{eq:2use}
\end{eqnarray}
since $I(\Theta_1; R_1)=I(\Theta_2; R_2)$ due to stationarity.
In other words, the difference in information transfer
when using two dependent or
independent inputs in the memoryless channel is equal to
$I(R_1; R_2)$. Obviously, for
$\Theta_1, \Theta_2$  independent holds $I(R_1; R_2)=0$.
The strength of the dependence between $R_1$ and
$R_2$ for correlated inputs
depends on the input range and the conditional response distributions,
see Fig.~\ref{fig:illu}. We expect $I(R_1; R_2)$ to be maximal
for the extreme input dependence, e.g.,
$\Theta_2=\Theta_1$, where $\Theta_1$ is equiprobably equal either to
$\theta_0-\Delta\theta$ or $\theta_0+\Delta\theta$. It follows, 
that  $R_1, R_2$ are conditionally (given $\Theta_1$)
identically and conditionally independently distributed. 
If $f(r|\theta_0-\Delta\theta)$ and $f(r|\theta_0+\Delta\theta)$ are
well separated, then $I(R_1; R_2)>0$ because $R_2$ provides redundant
information to $R_1$. As $\Delta\theta\rightarrow 0$, then 
$f(r|\theta_0-\Delta\theta)$ and $f(r|\theta_0+\Delta\theta)$ become
(almost) identical due to continuity in $\theta$ and thus 
$I(R_1; R_2)\rightarrow 0$. To make the argument precise, we show
that $I(R_1; R_2)= 0$ to the  second order in
the input amplitude, so that the effect of input correlations
in memoryless channels is of
higher order than the approximate Eq.~(\ref{eq:weakmi}).
The joint response distribution  is
\begin{eqnarray}
p(r_1, r_2)&=& 
\frac{1}{2} f(r_1| \theta_0 +\Delta\theta) f(r_2| \theta_0+ \Delta\theta)
+ \nonumber\\
&& +\frac{1}{2} f(r_1| \theta_0 -\Delta\theta) f(r_2| \theta_0- \Delta\theta),
\label{eq:rj}
\end{eqnarray} 
from which the marginals follow $p(r_1)=
f(r_1| \theta_0 +\Delta\theta)/2 + f(r_1| \theta_0 -\Delta\theta)/2$,
and similarly for $p(r_2)$. We employ another formula for MI
\citep[p.251]{r:cover} 
\begin{equation}
I(R_1; R_2)= \kl{p(r_1, r_2)}{p(r_1) p(r_2)}.
\label{eq:mirp}
\end{equation}
By substituting from Eq.~(\ref{eq:rj}) into 
Eq.~(\ref{eq:mirp}), and by employing the Taylor expansion 
in $\Delta\theta$ around $\Delta\theta=0$, we have (the terms up to
$\Delta\theta$ are zero)
\begin{equation*}
I(R_1; R_2)\approx (\Delta\theta)^2\!\!
\iint\limits_{R_1\times R_2} 
\left.\left[\pd{f(r_1| \theta)}{\theta} 
\pd{f(r_2|\theta)}{\theta}\right]\right|_{\theta=\theta_0}
\!\!\! dr_1\,dr_2,
\end{equation*}
which is equal to zero, due to Eq.~(\ref{eq:cdiff}).
The first nonzero term is of 4-th order,
and can be written as
$
(\Delta\theta)^4 J(\theta_0|R_1) J(\theta_0|R_2)/2
$,
provided that $f(r|\theta)$ is three times continuously
differentiable in $\theta$.

On the other hand, for channels with
memory the input correlations do matter, irrespectively of the
smallness of the amplitude. Consider, for example, two channel uses
in the additive noise case, $R_i= \Theta_i+ Z_i$,
$\aver{Z_i}{}=0$, where $i=1,2$. It is possible to approach the
noiseless channel in the extreme case of matching input and noise
correlations in accord with Eq.~(\ref{eq:weakmi}), e.g.,
if $\corr(Z_1, Z_2)\rightarrow -1$
and $\corr(\Theta_1, \Theta_2)\rightarrow 1$, 
then $R_1= \Theta_1+ Z_1$ and $R_2= \Theta_1 -Z_1$ and so
by adding $R_1 +R_2$ we can recover the value of
$\Theta_1$ perfectly.

\subsection{Small input power limit}

The signal power \citep{r:dorf}, $P_{\bm\Theta}$,
of an input signal described by  r.v. $\bm\Theta$ 
is defined as
\begin{equation}
P_{\bm\Theta}= \frac{1}{n}\aver{\bm\Theta^\T\bm\Theta}{}.
\end{equation}
For the covariance matrix $\mathbf C_{\bm\Theta}$
of r.v. $\bm\Theta$ holds
$
\mathbf C_{\bm\Theta}= 
\aver{(\bm\Theta -\aver{\bm\Theta}{})
(\bm\Theta -\aver{\bm\Theta}{})^\T}{}
$,
and therefore
\begin{equation}
P_{\bm\Theta}= \frac{1}{n}
\left[
\tr \mathbf C_{\bm\Theta} +
\|\aver{\bm\Theta}{}\|^2
\right].
\label{eq:power}
\end{equation}
The information channel is constrained in the input power $P$ if 
 only inputs that satisfy $P\geq P_{\bm\Theta}$ are considered.
It is common  in information theory of power-constrained
channels, to assume $\aver{\bm\Theta}{}=\mathbf 0$, 
then $P_{\bm\Theta}= \tr \mathbf C_{\bm\Theta}/n$
\citep[p.277]{r:cover}, which we  assume here also.
The assumption $\aver{\bm\Theta}{}=\mathbf 0$ results in simpler
notation, although it does not affect the generality of
results.
Due to stationarity, the marginal variances of r.v. $\bm\Theta$ are
constant, $\var(\Theta_i)=const.$ for all $i$, thus we can 
write\begin{equation}
\bm\Theta= \varepsilon \tilde{\bm\Theta},
\label{eq:thetadef}
\end{equation}
where $\var(\tilde{\Theta}_i)=1$ and
$\varepsilon>0$ is the scaling factor. The power of the input is then
$P_{\bm\Theta}=\varepsilon^2$, and the vanishing input power
is achieved by $\varepsilon\rightarrow 0$.

The approximate expression
for MI in the vanishing input power limit is
obtained analogously to the proof presented in Appendix~A, by
expressing $I(\bm\Theta; \mathbf R)$ in terms of the auxiliary r.v.
$\tilde{\bm\Theta}$, and then expanding for $\varepsilon\rightarrow 0$  
around $\varepsilon=0$. Let $\bm\Theta\sim \pi(\bm\theta)$ and 
$\tilde{\bm\Theta}\sim g(\tilde{\bm\theta})$, then 
from Eq.~(\ref{eq:thetadef}) follows $\pi(\bm\theta)= 
g(\bm\theta/\varepsilon)/\varepsilon =g(\tilde{\bm\theta})/\varepsilon$, 
and also $d\bm\theta= \varepsilon\,d\tilde{\bm\theta}$. 
The MI can be written by (analogously to Eq.~(\ref{eq:mutinfcm}))
\begin{equation}
I(\bm\Theta;\mathbf R)= 
\aver{\kl{f(\mathbf r| \varepsilon\tilde{\bm\theta})}
{\aver{f(\mathbf r| \varepsilon\tilde{\bm\theta})}{\tilde{\bm\theta}}}}
{\tilde{\bm\theta}}.
\end{equation}
The rest follows the argument of Appendix~A, although simplified due
to $\aver{\bm\Theta}{}=\mathbf 0$. It is obvious from the general proof,
that the assumption on zero $\aver{\bm\Theta}{}$ is not essential,
only that the vanishing input power is then with respect to
$\aver{\bm\Theta}{}$, so that $\tr \mathbf C_{\bm\Theta}/n$ is
the vanishing power of input fluctuations. Nevertheless, the
approximation is the same in both cases and reads
\begin{equation}
I(\bm\Theta; \mathbf R)\approx 
\frac{\varepsilon^2}{2}
\tr [\mathbf J(\bm\theta_0| \mathbf R) \mathbf C_{\tilde{\bm\Theta}}]
=
\frac{1}{2}
\tr [\mathbf J(\bm\theta_0| \mathbf R) \mathbf C_{\bm\Theta}],
\label{eq:mi-sp}
\end{equation}
where $\aver{\bm\Theta}{}=\bm\theta_0$.

Eqns.~(\ref{eq:weakmi}) and~(\ref{eq:mi-sp}) are identical, although
the assumptions on $\bm\Theta$ are different.
Consider for example the memoryless channel with power
constraint $P\geq\varepsilon^2$ on the input and $\aver{\Theta}{}=0$,
so that Eq.~(\ref{eq:mi-sp})
can be written as 
\begin{equation}
I(\Theta;R)\approx \frac{\varepsilon^2}{2} J(0| R).
\label{eq:memlpow}
\end{equation}
The capacity is achieved by
any distribution of inputs with power
$P_{\Theta}=\varepsilon^2=P$, for example
by the discrete distribution from
Eq.~(\ref{eq:bin}) with $\Delta\theta=\sqrt{P}$, or by  the
Gaussian distribution $\normal{0}{P}$.
Specifically, it is well known that the capacity of a power-constrained
linear additive white Gaussian noise (AWGN) channel is
\citep{r:cover}
\begin{equation}
\capa= \frac{1}{2} \ln \left(1+ \frac{P}{N}\right),
\label{eq:shannon}
\end{equation}
where $P$ is the power constraint on  the input and $N$ is the noise
power, and that the capacity is achieved by a normal distribution
$\normal{0}{P}$. The signal-to-noise ratio (SNR) is then defined as 
${\rm SNR}= P/N$.
By expanding Eq.~(\ref{eq:shannon}) to first order in $P$ for
$P\ll N$ we have $\capa\approx P/N/2$, which corresponds exactly
to Eq.~(\ref{eq:memlpow}), since for the Gaussian additive noise
holds $J(0|R)= 1/N$. A detailed review of AWGN 
channel capacity and its
different approximations for different SNR regimes 
(including the high-noise approximation above) 
can be found in
\cite{r:forney98}.
The conclusion
that in the vanishing input-power limit the capacity of AWGN channel
can be achieved by both discrete and $\normal{0}{P}$ distributions is
not so surprising in the light of some recent research on the AWGN
channels \citep{r:meyn}. It has been shown, that although the optimal
input distribution is generally $\normal{0}{P}$, the capacity can be
near-achieved by a discrete distribution, and specially, if $P\ll
N$ the other possible capacity-bearing 
distribution is indeed binary discrete. The methods
employed in \citep{r:meyn} are, however, different from our approach.
We further discuss the compatibility of Eq.~(\ref{eq:mi-sp}) with
the exact results obtained for non-white AGN channels in the low-input
power regime in the Results section of this paper.

\subsection{Simple lower bound on memoryless channel capacity}

We have demonstrated in the previous sections, that if the input to
the memoryless channel is weak (in amplitude or power), the optimal
distribution is discrete and binary. Therefore the channel capacity
cannot be more than $1$\,bit. Note, however, that the capacity can be
larger than $1$\,bit for channels with memory under certain
circumstances, as we demonstrate in the Results section.

It follows from the proof in Appendix~A, that
the Fisher information arises in Eq.~(\ref{eq:weakmi})
from Taylor-expanding the involved KL
distances in the expression for MI. More precise approximation to
channel capacity, $\capa_{\rm bin}$, can be thus obtained
without Taylor expansions, just by substituting
the discrete input distribution from Eq.~(\ref{eq:bin}) into
Eq.~(\ref{eq:mi}), 
\begin{eqnarray}
\capa_{\rm bin}&=& \frac{1}{2}\kl{f(r|\theta_0 -\Delta\theta)}{p(r)}
+\nonumber\\
&&+ \frac{1}{2}\kl{f(r|\theta_0 +\Delta\theta)}{p(r)},
\label{eq:bincap}
\end{eqnarray}
where  $p(r)= f(r|\theta_0 -\Delta\theta)/2 + f(r|\theta_0
+\Delta\theta)/2$. The parameter $\Delta\theta$ is half of the
maximum input amplitude for amplitude-constrained channels, 
and $\Delta\theta=\sqrt{P}$ for power-constrained channels.

Eq.~(\ref{eq:bincap}) is the lower bound on the true capacity,
$\capa\geq \capa_{\rm bin}$, which holds whether the amplitude (or
power) is small or not. The extension of Eq.~(\ref{eq:bincap}) to
channels with memory is not straightforward, for example 
the calculation of
$\capa_{\rm bin}$ would require numerical evaluation of possibly
high-dimensional integrals which may not be numerically stable
\citep{r:coxreid}. Therefore for channels with memory we propose to
employ Eq.~(\ref{eq:weakmi}) as the simplest method.

\section{Results for selected systems}

\subsection{Memoryless channels}

\subsubsection{Amplitude constrained linear AWGN channel}

The capacity and capacity-bearing input distributions of
the linear AWGN channel,
\begin{equation}
R= \Theta + Z,
\end{equation}
where r.v. $Z$ is zero-mean Gaussian and the input is constrained
in amplitude, 
were studied in detail in \citep{r:smith71}. Contrary to the well
known Eq.~(\ref{eq:shannon}) for the input power constrained channel,
no closed-form expression for capacity exists in the amplitude
constrained version, moreover the optimal input distribution is known
to be discrete with finite set of mass points.

We assume  $\theta_0=0$, the maximal input amplitude is
$2\Delta\theta$, thus the input is bound to lie in the interval
$[-\Delta\theta, \Delta\theta]$. Furthermore we assume
that the power of the noise is $N=1$, so the noise is described by the
standard normal r.v., $Z\sim\normal{0}{1}$.
Eq.~(\ref{eq:capamemless}) then becomes
\begin{equation}
\capa_{\rm high}= \frac{1}{2} (\Delta\theta)^2.
\end{equation}
The binary approximation, $\capa_{\rm bin}$ given by Eq.~(\ref{eq:bincap}),
has to be evaluated numerically.  Additionally, we also investigate
the \emph{low noise} approximation to MI, $\capa_{\rm low}$,
which is also based on FI
\citep{r:bernardo79, r:brunelnadal, r:mcdonnell08prl},
\begin{eqnarray}
\capa_{\rm low}= \ln 
\frac{\int_{\Theta} \sqrt{J(\theta|R)}\,d\theta }
{\sqrt{2\pi e}}.
\label{eq:capalow}
\end{eqnarray}
Eq.~(\ref{eq:capalow}) is a lower bound on the true channel capacity,
$\capa\geq\capa_{\rm low}$,  tight with the vanishing
noise in the information transmission. In the case of
amplitude-constrained AWGN channel we have
\begin{equation}
\capa_{\rm low}= \ln \frac{2\Delta\theta}{\sqrt{2\pi e}}.
\end{equation}

Fig.~\ref{fig:memless}a. shows
the comparison of the exact channel capacity (data taken from
\citep{r:dauwels}) with $\capa_{\rm high}, \capa_{\rm bin}$ and
$\capa_{\rm low}$, expressed as functions of the
signal-to-noise ratio (in dB), which is defined as
\citep{r:dauwels}
\begin{equation}
{\rm SNR}= 10 \log_{10} \left[(\Delta\theta)^2\right].
\end{equation}
The capacities are evaluated in bits  which means
converting the natural logarithms in Eqns.~(\ref{eq:capamemless}),
~(\ref{eq:bincap}) and~(\ref{eq:capalow}) to base~2, i.e., to divide
the values by $\ln 2$.
While $\capa_{\rm low}$ and $\capa_{\rm high}$ provide good
approximations only for rather high and small SNR values, the
$\capa_{\rm bin}$ approximation gives good results even for
intermediate SNR values. A similar figure with additional
approximations for the classical AWGN channel capacity can be found in
\cite{r:forney98}.

\begin{figure}
\begin{center}
\includegraphics[]{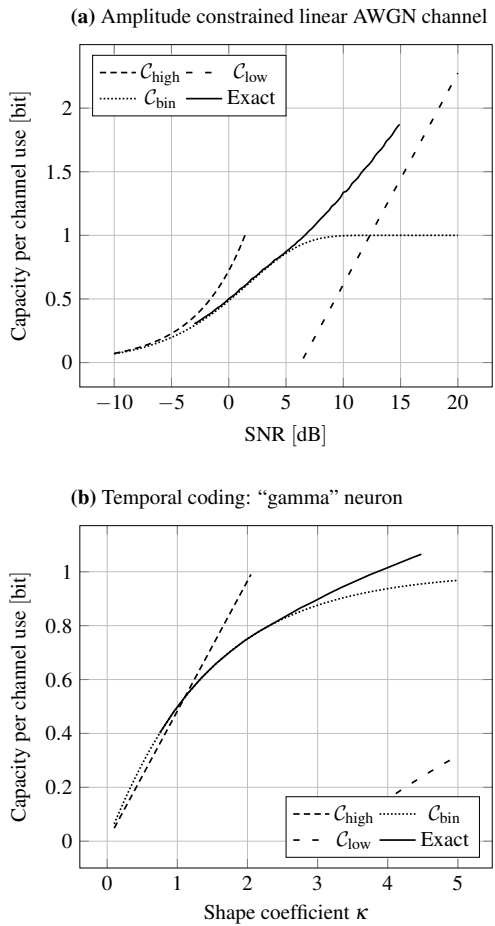}
\end{center}
\newpage
\caption{Capacities and their approximations in memoryless channels.
The high-noise capacity approximation ($\capa_{\rm high}$,
Eq.~(\ref{eq:capamemless})) approximates the true capacity of the
amplitude-constrained AWGN channel \textbf{(a)} well only for very
low signal-to-noise ratios (SNR),
just like the low-noise approximation 
($\capa_{\rm low}$, Eq.~(\ref{eq:capalow})) 
does for high SNRs. The
binary-channel approximation
($\capa_{\rm bin}$, Eq.~(\ref{eq:bincap}))
holds well even for  intermediate-low SNRs.
The exact solution is taken from \cite{r:dauwels}.
The information capacity of a simple model of neuronal
coding \textbf{(b)} apparently falls into the high-noise category,
since both $\capa_{\rm high}$ and $\capa_{\rm bin}$ approximate 
the true capacity (taken from \cite{r:ikeda}) better than
$\capa_{\rm low}$.
}
\label{fig:memless}
\end{figure}

\subsubsection{Temporal neuronal coding}

Recently, the information capacity of a memoryless
neuronal model has been
analyzed in detail \citep{r:ikeda}. It is assumed, that
the neuronal response $R$ is the interval between
two consequent action potentials. In agreement with some experimental
observations \citep{r:levine, r:mckeegan02, r:pouzatstar, r:reeke},
the response for each input follows  the gamma distribution,
\begin{equation}
f(r|\theta)= \frac{r^{\kappa-1}}{\theta^\kappa}
\frac{\exp(-r/\theta)}{\Gamma(\kappa)},
\end{equation}
where the  parameter $\theta$ is assumed to be the input (stimulus
intensity).
Based on further experimental observations \citep{r:shinomoto03},
the input is constrained in amplitude,
$5/\kappa\leq \theta\leq 50/\kappa$. The exact
capacity was calculated  numerically by 
\citet{r:ikeda}
 for $0.75\leq \kappa\leq 4.5$.

While $\capa_{\rm bin}$ has to be evaluated numerically, for the high
and low noise approximations we have
\begin{align}
\capa_{\rm high} &= \frac{81}{242}\kappa,&
\capa_{\rm low} &= \ln \frac{\sqrt{\kappa}\ln 10}{\sqrt{2\pi e}}.
\end{align}
The results are shown in Fig.~\ref{fig:memless}b. 
For the investigated values of $\kappa$, both
$\capa_{\rm high}$ and $\capa_{\rm bin}$ approximations give better
results than $\capa_{\rm low}$, which suggests that this particular 
case of temporal coding falls within the ``high noise'' category. 
Neuronal responses often vary substantially across identical
stimulus trials, thus the highly noisy information transmission
is not unusual as reported from experimental measurements 
\citep{r:carandini}. A simple model of a stochastic resonance
in an electrosensory neuron,
subject to sub-threshold (i.e., very weak)
stimulation \citep{r:greenwoodprl,
r:langreen-bc} has been analyzed  by employing $\capa_{\rm high}$
recently \citep{r:koslan10prer}.

\subsection{Linear Gaussian channel with memory and input power
constraint}

First, we demonstrate that Eq.~(\ref{eq:mi-sp}) is compatible
with exact results available on input power constrained 
linear AGN channels with memory \citep{r:cover,
r:yeung} in the limit
of weak input power. The channel is defined as
\begin{equation}
\mathbf R= \bm\Theta +\mathbf Z,
\label{eq:memagn}
\end{equation}
where the zero-mean input is constrained in power $P$ 
\citep[p.277]{r:cover}, 
\begin{equation}
P\geq \frac{1}{n} \tr \mathbf C_{\bm\Theta},
\label{eq:agnpow}
\end{equation}
and the noise 
is given by the multivariate normal distribution with covariance
matrix $\mathbf C_{\mathbf Z}$, $\mathbf Z\sim \normal{\mathbf
0}{\mathbf C_{\mathbf Z}}$. 
The channel conditional p.d.f. is
therefore
\begin{equation}
f(\mathbf r|\bm\theta) 
=
\frac{1}{\sqrt{(2\pi)^n\det \mathbf C_{\mathbf Z}}}
\exp\left[
(\mathbf r-\bm\theta)^{\T}
\mathbf C_{\mathbf Z}^{-1}
(\mathbf r-\bm\theta)
\right],
\label{eq:mvnormal}
\end{equation}
and substituting Eq.~(\ref{eq:mvnormal}) into Eq.~(\ref{eq:fi}) gives
\citep{r:kay}
\begin{equation}
\mathbf J(\bm\theta|\mathbf R)=\mathbf C_{\mathbf Z}^{-1},
\label{eq:fiagn}
\end{equation}
which is independent of $\bm\theta$.

From the spectral decomposition theorem \citep{r:horn} follows
that
\begin{equation}
\mathbf C_{\mathbf Z}= \mathbf Q\bm\Lambda\mathbf Q^{\T},
\label{eq:czdiag}
\end{equation}
where the matrix $\bm\Lambda$ is diagonal with positive elements
and $\mathbf Q$ is
orthonormal. The capacity per channel use 
is then given by \citep{r:yeung}
\begin{equation}
\capa= \frac{1}{2n} \sum_{i=1}^n
\ln \left(1+\frac{m_{i}}{[\bm\Lambda]_{ii}}\right),
\label{eq:agncap}
\end{equation}
where the constants $m_i\geq 0$  are determined by the
water-filling procedure \citep[p.274]{r:cover}, so
that the power constraint given by Eq.~(\ref{eq:agnpow}) holds as
$
\sum_{i=1}^n m_i= nP
$.
Furthermore, the optimal input distribution  is also multivariate
normal, $\bm\Theta\sim \normal{\mathbf 0}{\mathbf C_{\bm\Theta}}$,
with covariance matrix $\mathbf C_{\bm\Theta}= \mathbf Q\mathbf
M\mathbf Q^\T$ \citep[p.279]{r:yeung}, where the diagonal matrix
$\mathbf M$ is defined as $[\mathbf M]_{ii}=m_i$.

In order to obtain the vanishing input power limit of
Eq.~(\ref{eq:agncap}), we observe that as  $P\rightarrow 0$ 
also $m_i\rightarrow 0$, so we can expand Eq.~(\ref{eq:agncap})
as
\begin{equation}
\capa\approx \frac{1}{2n} \sum_{i=1}^n \frac{m_i}{[\bm\Lambda]_{ii}}
= \frac{1}{2n} \tr\left(\bm\Lambda^{-1}\mathbf M\right).
\label{eq:agnapprox}
\end{equation}
By combining Eqns.~(\ref{eq:fiagn}), (\ref{eq:czdiag}), 
(\ref{eq:agnapprox})
and basic properties of matrix inverse and trace \citep{r:horn}
we have
\begin{eqnarray}
\capa&\approx& 
\frac{1}{2n}
\tr[(\mathbf Q^{\T}\mathbf C_{\mathbf Z}\mathbf Q)^{-1}\mathbf M]=
\frac{1}{2n}
\tr[\mathbf Q^{\T}\mathbf C^{-1}_{\mathbf Z}\mathbf Q\mathbf M]=
\nonumber\\
&=&
\frac{1}{2n}
\tr[\mathbf C^{-1}_{\mathbf Z}\mathbf Q\mathbf M\mathbf Q^\T]=
\frac{1}{2n}
\tr[\mathbf J(\bm\theta| \mathbf R) \mathbf C_{\bm\Theta}],
\label{eq:agncapappr}
\end{eqnarray}
which corresponds to the capacity per channel use
as $n\rightarrow \infty$, due to Eq.~(\ref{eq:mi-sp}), 
for power achieving input,
$\tr \mathbf C_{\bm\Theta}/n=P$.

Next, we  illustrate Eq.~(\ref{eq:agncapappr}) on two simple models
of Gaussian noise with memory.

\subsubsection{AR(1) noise}

The channel is given by Eqns.~(\ref{eq:memagn}) and~(\ref{eq:agnpow}),
with $Z_i$'s following the AR(1) process: $Z_i= \varrho Z_{i-1}
+X_i$, where $-1<\varrho <1$ is the  correlation coefficient,
$\varrho=\corr(Z_i, Z_{i-1})$, and $X_i$ are independently
distributed standard normal r.v.'s, $X_i \stackrel{\rm i.i.d.}{\sim}
\normal{0}{1}$ \citep{r:kendall3}. 
The noise covariance matrix  has elements
\begin{equation}
[\mathbf C_{\mathbf Z}]_{ik}= \varrho^{|i-k|},
\end{equation}
and its inverse, equal to the FI matrix by Eq.~(\ref{eq:fiagn}),
is tridiagonal,
\begin{equation}
\mathbf J(\bm\theta|\mathbf R)= \frac{1}{1-\varrho^2}
\left(\begin{array}{cccccc}
1 & -\varrho & 0 &  \cdots & 0\\
-\varrho & 1+\varrho^2 & -\varrho & \cdots & 0\\
0 &  -\varrho & \ddots & \ddots & \vdots \\ 
\vdots & \vdots & \ddots & 1+\varrho^2  & -\varrho \\ 
0 & 0 & 0 & -\varrho & 1
\end{array}\right).
\label{eq:fiar1}
\end{equation}
We denote the correlation coefficient between consequent inputs as
$c=\corr(\Theta_i, \Theta_{i+1})$. The MI per channel use
for maximum power achieving input, $P=\tr \mathbf C_{\bm\Theta}/n$,
can be found exactly by employing Eq.~(\ref{eq:mi-sp}),
\begin{eqnarray}
\lim_{n\rightarrow\infty} \frac{1}{n} I(\bm\Theta; \mathbf R) =
\frac{P}{2} \frac{\varrho^2 +1 -2c\varrho}{1-\varrho^2}.
\label{eq:miar2}
\end{eqnarray}
For $\varrho=0$ (memoryless channel)
the value of $c$ does not matter as discussed earlier.
The capacity per channel use is 
\begin{equation}
\capa_{\rm high} 
= \frac{P}{2} \frac{\varrho^2 +1 +2 |\varrho|}{1-\varrho^2},
\label{eq:miarcap}
\end{equation}
since $\sup_{-1< c< 1}(-c\varrho)= |\varrho|$.
The capacity in bits per vanishing input power, $\capa_{\rm high}/P$,
is shown in Fig.~\ref{fig:memcap} in dependence on the noise
correlation $\varrho$. Note that from Eq.~(\ref{eq:miarcap}) follows
$\capa_{\rm high}/P\rightarrow \infty$ as $|\varrho|\rightarrow 1$,
i.e., as the noise correlation increases, its corrupting power
decreases and in the limit we can approach the noiseless channel.

\begin{figure}[]
\begin{center}
\includegraphics[]{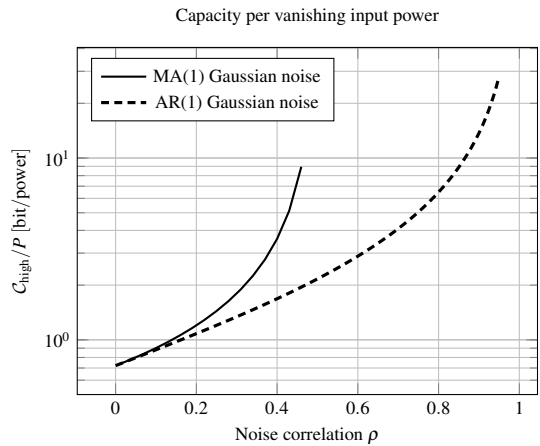}
\end{center}
\caption{The capacities per vanishing input powers
for the AR(1) and MA(1) Gaussian additive noise models  
in dependence on the the noise correlation coefficient
$\varrho$ (the graphs are symmetric in $\varrho$). 
Note that  the capacity tends to infinity as
$|\varrho|\rightarrow 0.5$ (the MA(1) model)
and as $|\varrho|\rightarrow 1$ (the AR(1) model). In these limits,
the corrupting power of the noise in the information transfer is
decreased to the point, that the channel approaches the noiseless
channel and the input value can be recovered perfectly.
}
\label{fig:memcap}
\end{figure}

\subsubsection{MA(1) noise}

The channel is given by Eqns.~(\ref{eq:memagn}) and~(\ref{eq:agnpow}),
r.v.'s $Z_i$ follow the MA(1) process, $Z_i= X_i -\gamma X_{i-1}$,
where $-1< \gamma <1$ is the parameter of the process and $X_i
\stackrel{\rm i.i.d.}{\sim} \normal{0}{1}$.  The parameter of the
MA(1) process and the correlation coefficient
$\varrho=\corr(Z_i, Z_{i-1})$ are related as $\varrho=
-\gamma/(1+\gamma^2)$, and therefore $-0.5< \varrho < 0.5$
\citep{r:kendall3}.  The covariance matrix of the MA(1) process is
tridiagonal, and its inverse has all elements non-zero, although
decreasing in absolute value with the distance from the main 
diagonal, see Fig.~\ref{fig:mem}a,~b.

Recently, a closed form expression for $\mathbf
C^{-1}_{\mathbf Z}$ of the MA(1) process has been published
\citep{r:kumar}. The expression is rather complicated and we
cannot evaluate the analogous limit to Eq.~(\ref{eq:miar2}) in a
closed form. Nevertheless, we approximate the capacity per
channel use by considering $n$ high enough, and the closed form
expression for the elements of the FI matrix allows us
to avoid numerical issues when inverting the covariance
matrix. 
The capacity per vanishing input power, $\capa_{\rm high}/P$,
is shown in Fig.~\ref{fig:memcap}. Note, that for $n\leq 2000$ we were
unable to obtain stable values of $\capa_{\rm high}$ for 
$|\varrho|>4.2$. This is caused by the fact, that the dominant terms
of the FI matrix, and consequently $\capa_{\rm high}/P$,
diverge to $+\infty$ as $|\varrho|\rightarrow
0.5$ (in a similar way as Eq.~(\ref{eq:miarcap}) does for 
$|\varrho|\rightarrow 1$). In other words, the dependence structure of
the MA(1) process is sufficiently ``rigid'' even for intermediate
correlation values, that by properly matching the input correlations
we can  approach the noiseless information transfer. 
The examples of optimally matched input signals are shown in
Fig.~\ref{fig:mem}c,~d,~e.

\begin{figure}[]
\begin{center}
\includegraphics[]{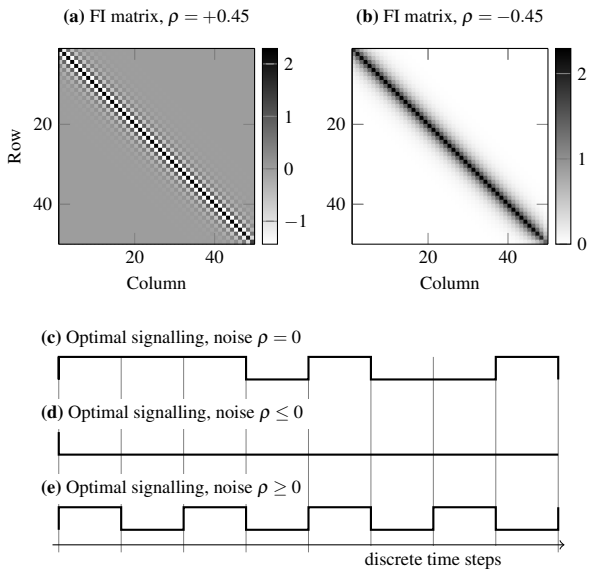}
\end{center}
\caption{Small input amplitude optimality conditions for 
linear channels with AR(1) or MA(1) additive Gaussian noise.
The  structure of the 
Fisher information matrix of the MA(1) model (panels
\textbf{(a)}
and \textbf{(b)} for $n=50$) shows elements decaying in 
absolute value with
distance from the main diagonal, sign changes occur for positively
correlated MA(1) process and all elements are positive for
$\varrho\leq 0$. The structure of the FI matrix determines the
covariance matrix of the optimal signal. Panel \textbf{(c)} shows the
example optimal input to the memoryless channel 
(noise  correlation $\varrho=0$): random switching
between input values $+\sqrt{P}$ and
$-\sqrt{P}$ (discrete binary input), where $P$ is the
input power constraint.
The same capacity would be achieved by input values described
by the normal distribution $\normal{0}{P}$, as discussed in the text.
Depending on the sign of the
noise correlation $\varrho$, the optimal input is characterized
by extremal value of correlation between consequent inputs (panels 
\textbf{(d)} and \textbf{(e)}). Note, that the capacity of the
memoryless channel is achieved by \textbf{(d)} and \textbf{(e)} also,
independently on the input correlations.
}
\label{fig:mem}
\end{figure}

\section{Conclusions}
We derive  approximate expression for mutual information 
in a broad class of discrete-time stationary channels
(including those with memory) with continuous, but small, input. The
input is restricted either in amplitude or in power and we study the
optimality conditions on information transfer as the power or
amplitude approach zero.  We find that the input and channel
properties are separated in the approximate formula, which allows us
to study the optimality conditions in a convenient way. Specifically,
we find that the increase of mutual information from zero power (or
amplitude) for a given channel depends only on the input covariances.

For memoryless channels, the capacity cannot be more than $1$\,bit per
channel use and the optimal input is  unique discrete binary
distribution in the small input amplitude case, but generally
non-unique in the small input power case. We demonstrate, that the
effect of input correlations in memoryless channels is of higher order
than the order of the capacity approximation, and thus the additional
correlations do not decrease the capacity although they decrease the
input entropy. We also provide a simple lower bound on capacity of
memoryless channels subject to weak-stimulus constraints that gives
better results in practical situations. 

In channels with memory, the capacity can be greater than $1$\,bit and
the input correlations play the most important role. We show, that the 
approximate formula includes the small input power limit of the exact
solution for linear additive Gaussian noise channels with
memory. We show, that by properly matching the input covariances to
the dependence structure of the noise, we can approach in certain
cases the noiseless channel even for intermediate values of the noise
correlations.

\begin{acknowledgments}
This work was supported by AV0Z50110509 and  Centre for Neuroscience
LC554. I thank Ales Nekvinda for helpful comments on the Appendix.
\end{acknowledgments}

\appendix
\section{Capacity in the vanishing input amplitude}

We introduce an auxiliary r.v. $\delta\bm\Theta$ 
by employing Eq.~(\ref{eq:thetares}) as
\begin{equation}
\delta\bm\Theta= \bm\Theta- \bm\theta_0,
\label{eq:dtheta}
\end{equation}
so that for all $i$ holds $\delta\theta_i\in [-\Delta\theta,
\Delta\theta]$. The p.d.f. of r.v.
$\delta\bm\Theta$ is denoted as $\pi(\delta\bm\theta)$.
Mutual information $I(\bm\Theta;\mathbf R)$ 
from Eq.~(\ref{eq:mi}) can be  written in terms of r.v. 
$\delta\bm\Theta$,
whether $\|\Delta\bm\theta\|$ is small or not as
\begin{equation}
I(\bm\Theta;\mathbf R)=
\aver{\kl{f(\mathbf r|\bm\theta_0+\delta\bm\theta)}
{\aver{f(\mathbf r|\bm\theta_0+ \delta\bm\theta)}{\delta\bm\theta}}
}{\delta\bm\theta}.
\label{eq:mutinfcm}
\end{equation}
In order to approximate
$I(\bm\Theta;\mathbf R)$   around
$\bm\theta_0$ in terms of $\delta\bm\theta$ for small $\|\Delta\bm\theta\|$,
we need to expand the KL distance in
Eq.~(\ref{eq:mutinfcm}). We introduce 
\begin{eqnarray}
\varphi(\mathbf r, \bm\theta_0+\delta\bm\theta)&=&
f(\mathbf r|\bm\theta_0+\delta\bm\theta) \ln
f(\mathbf r|\bm\theta_0+\delta\bm\theta),
\\
\psi(\mathbf r, \bm\theta_0+\delta\bm\theta)&=&
f(\mathbf r|\bm\theta_0+\delta\bm\theta) \ln
\aver{f(\mathbf r|\bm\theta_0+ \delta\bm\theta)}{\delta\bm\theta},
\end{eqnarray}
and rewrite the KL distance as
\begin{multline}
\kl{f(\mathbf r|\bm\theta_0+\delta\bm\theta)}
{\aver{f(\mathbf r|\bm\theta_0+ \delta\bm\theta)}{\delta\bm\theta}}
=\\
=
\int_{\mathbf R}
\left[
\varphi(\mathbf r, \bm\theta_0+\delta\bm\theta)
- \psi(\mathbf r, \bm\theta_0+\delta\bm\theta)
\right]\,d\mathbf r,
\label{eq:klpp}
\end{multline}
thus reducing the problem to expanding $\varphi(\mathbf r, \bm\theta)$
and $\psi(\mathbf r, \bm\theta)$. While the
Taylor expansion of $\varphi(\mathbf r, \bm\theta)$
is straightforward, the expansion of the
logarithm of the expected value of $f(\mathbf r| \bm\theta)$ in
$\psi(\mathbf r, \bm\theta)$ is examined in the following Lemma.

\newtheorem{lem}{Lemma}
\begin{lem}\label{lem}
Let  $f(\mathbf r| \bm\theta)$ be twice
continuously differentiable 
with respect to $\bm\theta$. Then
for a chosen $\bm\theta_0$, r.v. $\delta\bm\Theta\sim
\pi(\delta\bm\theta)$ and $\Delta\bm\theta$ 
such,  that for all $i$ holds $\Delta\theta>0$ and
$-\Delta\theta\leq \delta\theta_i \leq \Delta\theta$,
there
exists $P>0$ such, that the following approximation 
for small enough $\|\Delta\bm\theta\|$
holds
\begin{equation}
\ln
\aver{f(\mathbf r|\bm\theta_0+ \delta\bm\theta)}{\delta\bm\theta}
\approx \ln f(\mathbf r| \bm\theta_0) 
+ \aver{\delta\bm\Theta}{}^\T 
\frac{\nabla f(\mathbf r| \bm\theta_0)}{f(\mathbf r| \bm\theta_0)},
\label{eq:celkem}
\end{equation}
where  $\nabla f(\mathbf r| \bm\theta_0)=\left.
\nabla f(\mathbf r| \bm\theta)\right|_{\bm\theta=\bm\theta_0}$,
the gradient is taken with respect to $\bm\theta$ and
$\aver{\delta\bm\Theta}{}= \aver{\delta\bm\Theta}{\delta\bm\theta}$
is the expectation of r.v. $\delta\bm\Theta$.
The maximum error of expansion~(\ref{eq:celkem}) is
bounded by $P\|\Delta\bm\theta\|^2$. 
\end{lem}

\begin{proof}
From the continuity of second derivatives of $f(\mathbf r| \bm\theta)$
around $\bm\theta_0$ follows
\begin{equation}
\left|\pd{^2 f(\mathbf r|\bm\theta)}{\theta_i \,\partial\theta_j}\right|
\leq M,
\end{equation}
for all $i,j$.
The Taylor expansion of $f(\mathbf r|\bm\theta)$ around $\bm\theta_0$
in terms of $\delta\bm\theta$ reads
\begin{equation}
f(\mathbf r| \bm\theta_0 + \delta\bm\theta)
\approx
f(\mathbf r| \bm\theta_0) +\delta\bm\theta^{\T}
\nabla f(\mathbf r| \bm\theta_0),
\label{eq:tayl1}
\end{equation}
and furthermore
\begin{multline}
\left|
f(\mathbf r| \bm\theta_0 + \delta\bm\theta)
- f(\mathbf r| \bm\theta_0) -\delta\bm\theta^{\T}
\nabla f(\mathbf r| \bm\theta_0)
\right|
\leq \\
\leq nM \|\delta\bm\theta\|^2 \leq C\|\Delta\bm\theta\|^2.
\label{eq:tayl1err}
\end{multline}
By integrating the expansion~(\ref{eq:tayl1}), i.e., 
by taking the expectation with respect to r.v.
$\delta\bm\Theta$, and by employing inequality~(\ref{eq:tayl1err})
it can be established that
\begin{widetext}
\begin{multline}
\left|
\int_{\mathbf R} \pi(\delta\bm\theta)
f(\mathbf r| \bm\theta_0 + \delta\bm\theta)
\,d(\delta\bm\theta)
- f(\mathbf r| \bm\theta_0) 
- \aver{\delta\bm\Theta}{}^\T\nabla f(\mathbf r| \bm\theta_0) 
\right| 
= \\
= \left|
\int_{\mathbf R} \pi(\delta\bm\theta)
\left[
f(\mathbf r| \bm\theta_0 + \delta\bm\theta)
- f(\mathbf r| \bm\theta_0) -\delta\bm\theta^{\T}
\nabla f(\mathbf r| \bm\theta_0) \right]
\,d(\delta\bm\theta)
\right|  \leq\int_{\mathbf R} \pi(\delta\bm\theta)
C\|\Delta\bm\theta\|^2
\,d(\delta\bm\theta) =C\|\Delta\bm\theta\|^2,
\label{eq:tayl2err}
\end{multline}
and therefore the following expansion holds
\begin{equation}
\int_{\mathbf R} \pi(\delta\bm\theta)
f(\mathbf r| \bm\theta_0 + \delta\bm\theta)
\,d(\delta\bm\theta)
\approx
f(\mathbf r| \bm\theta_0) 
+ \aver{\delta\bm\Theta}{}^\T\nabla f(\mathbf r| \bm\theta_0),
\label{eq:tayl2}
\end{equation}
with the maximum error of order $\|\Delta\bm\theta\|^2$.
From the Lagrange mean value theorem follows, that for 
$A, B>0$ holds
\begin{equation}
|\ln A -\ln B|\leq \frac{1}{\min(A,B)} |A-B|.
\label{eq:lmeanv}
\end{equation}
We set
$
A = \int_{\mathbf R} \pi(\delta\bm\theta)
f(\mathbf r| \bm\theta_0 + \delta\bm\theta)
\,d(\delta\bm\theta),
B = f(\mathbf r| \bm\theta_0) 
+ \aver{\delta\bm\Theta}{}^\T\nabla f(\mathbf r| \bm\theta_0),
$
and combine the inequalities~(\ref{eq:tayl2err}) and~(\ref{eq:lmeanv})
to obtain
\begin{eqnarray}
|\ln A -\ln B|
&=& \bigg|
\ln\int_{\mathbf R} \pi(\delta\bm\theta)
f(\mathbf r| \bm\theta_0 + \delta\bm\theta)
\,d(\delta\bm\theta) - \ln\left[
f(\mathbf r| \bm\theta_0) 
+ \aver{\delta\bm\Theta}{}^\T\nabla f(\mathbf r| \bm\theta_0) 
\right]
\bigg| \leq\nonumber\\ 
&\leq& \frac{1}{\min(A,B)} |A-B|
\leq \frac{1}{\min(A,B)} C\|\Delta\bm\theta\|^2,
\label{eq:p1}
\end{eqnarray}
where $\min(A,B)$ is finite due to regularity of $f(\mathbf r|
\bm\theta)$.
From the Taylor expansion of $\ln(a+x)$ around $a$ in terms of $x$ and
the expression for the Lagrange remainder \citep{r:abramowitz} we have
\begin{equation}
\left|
\ln(a+x) - \ln(a) -\frac{x}{a}
\right|\leq \frac{x^2}{a^2}.
\end{equation}
Setting $a= f(\mathbf r| \bm\theta_0)$ and
$x= \aver{\delta\bm\Theta}{}^\T \nabla f(\mathbf r|\bm\theta_0)$ thus
gives
\begin{equation}
\left|
\ln \left[
f(\mathbf r| \bm\theta_0) + 
\aver{\delta\bm\Theta}{}^\T \nabla f(\mathbf r|\bm\theta_0)
\right]
- \ln f(\mathbf r| \bm\theta_0) 
- \frac{\aver{\delta\bm\Theta}{}^\T 
	\nabla f(\mathbf r|\bm\theta_0)}{f(\mathbf r| \bm\theta_0) }
\right| 
\leq 
\frac{\|\nabla f(\mathbf r|\bm\theta_0)\|^2}{f^2(\mathbf r| \bm\theta_0)}
\|\Delta\bm\theta\|^2.
\label{eq:p2}
\end{equation}
Finally, we apply the triangle inequality for absolute value,
$
|\alpha -\beta|\leq |\alpha -\gamma| + |\gamma -\beta|,
$
setting
\begin{align}
\alpha&= \ln A = \ln \int_{\mathbf R} \pi(\delta\bm\theta)
f(\mathbf r| \bm\theta_0 + \delta\bm\theta)
\,d(\delta\bm\theta),& 
\beta &= \ln f(\mathbf r| \bm\theta_0) + 
	\frac{\aver{\delta\bm\Theta}{}^\T 
	\nabla f(\mathbf r|\bm\theta_0)}{f(\mathbf r| \bm\theta_0) },& 
\\
\gamma&= \ln B = \ln\left[
f(\mathbf r| \bm\theta_0) 
+ \aver{\delta\bm\Theta}{}^\T\nabla f(\mathbf r| \bm\theta_0)
\right],&
\end{align}
and by combining inequalities~(\ref{eq:p1}) and~(\ref{eq:p2})
we obtain
\begin{multline}
\left|
\ln \int_{\mathbf R} \pi(\delta\bm\theta)
f(\mathbf r| \bm\theta_0 + \delta\bm\theta)
\,d(\delta\bm\theta)
- \ln f(\mathbf r| \bm\theta_0) - 
	\frac{\aver{\delta\bm\Theta}{}^\T 
	\nabla f(\mathbf r|\bm\theta_0)}{f(\mathbf r| \bm\theta_0)}
\right| \leq\\
\leq 
\left|
\ln \int_{\mathbf R} \pi(\delta\bm\theta)
f(\mathbf r| \bm\theta_0 + \delta\bm\theta)
\,d(\delta\bm\theta)
- \ln\left[
f(\mathbf r| \bm\theta_0) 
+ \aver{\delta\bm\Theta}{}^\T\nabla f(\mathbf r| \bm\theta_0)
\right]
\right| + \\ 
+\left|
\ln\left[
f(\mathbf r| \bm\theta_0) 
+ \aver{\delta\bm\Theta}{}^\T\nabla f(\mathbf r| \bm\theta_0)
\right]
- \ln f(\mathbf r| \bm\theta_0) - 
	\frac{\aver{\delta\bm\Theta}{}^\T 
	\nabla f(\mathbf r|\bm\theta_0)}{f(\mathbf r| \bm\theta_0)}
\right| \leq\\
\leq
\frac{1}{\min(A,B)} C\|\Delta\bm\theta\|^2
+\frac{\|\nabla f(\mathbf r|\bm\theta_0)\|^2}{f^2(\mathbf r| \bm\theta_0)}
\|\Delta\bm\theta\|^2= P\|\Delta\bm\theta\|^2,
\end{multline}
\end{widetext}
and therefore
\begin{equation}
\ln\aver{f(\mathbf r| \bm\theta_0 + \delta\bm\theta)}{\delta\bm\theta}
\approx
\ln f(\mathbf r| \bm\theta_0) +\aver{\delta\bm\Theta}{}^{\T}
\frac{\nabla f(\mathbf r| \bm\theta_0)}{f(\mathbf r| \bm\theta_0)},
\end{equation}
with error of order $\|\Delta\bm\theta\|^2$.
\end{proof}

In the following we set
$\varphi\equiv\varphi(\mathbf r, \bm\theta_0+\delta\bm\theta)$, 
$\psi\equiv\psi(\mathbf r, \bm\theta_0+\delta\bm\theta)$,
$f\equiv f(\mathbf r|\bm\theta_0)$ and
$\nabla f\equiv \left.\nabla 
f(\mathbf r|\bm\theta)\right|_{\bm\theta= \bm\theta_0}$ for shorthand,
and by repeatedly applying Lemma~\ref{lem} and 
keeping in mind the rules for derivatives $(fg)''= f''g+2 f'g' +fg''$,
and $(\ln f)''= f''/f- (f'/f)^2$, we obtain the expansions
\begin{eqnarray}
\varphi &\approx&
f\ln f + \delta\bm\theta^{\T} \ln f \nabla f
+ \delta\bm\theta^{\T}\nabla f +\nonumber \\
&& +\frac{1}{2} \delta\bm\theta^{\T} \ln f
\nabla\nabla^{\T} f\,\delta\bm\theta + 
\delta\bm\theta^{\T} \frac{\nabla f\nabla^{\T} f}{f}\delta\bm\theta
+ \nonumber\\
&&+\frac{1}{2} \delta\bm\theta^{\T} f \left[
\frac{\nabla\nabla^{\T} f}{f} -\frac{\nabla f\nabla^{\T} f}{f^2}
\right]\delta\bm\theta,
\\
\psi &\approx&
f\ln f +\delta\bm\theta^{\T} \ln f \nabla f 
+\aver{\delta\bm\Theta}{}^{\T} \nabla f +\nonumber\\
&&+\frac{1}{2} \delta\bm\theta^{\T}\ln f \nabla\nabla^{\T}
f\,\delta\bm\theta
+ \delta\bm\theta^{\T} \frac{\nabla f\nabla^{\T}
f}{f}\aver{\delta\bm\Theta}{}+ \nonumber\\
&&+\frac{1}{2} \aver{\delta\bm\Theta}{}^{\T} f \left[
\frac{\nabla\nabla^{\T} f}{f} -\frac{\nabla f\nabla^{\T} f}{f^2}
\right] \aver{\delta\bm\Theta}{}.
\end{eqnarray}
We substitute these expansions into Eq.~(\ref{eq:klpp}), and by applying 
the regularity conditions~(\ref{eq:cdiff}) we have
\begin{multline}
\int_{\mathbf R}[\varphi-\psi]\,d\mathbf r
\approx
\frac{1}{2} \delta\bm\theta^{\T}
\mathbf J(\bm\theta_0|\mathbf R)
\delta\bm\theta 
-\\
-\delta\bm\theta^{\T}
\mathbf J(\bm\theta_0|\mathbf R)
\aver{\delta\bm\Theta}{}
+\frac{1}{2} \aver{\delta\bm\Theta}{}^{\T}
\mathbf J(\bm\theta_0|\mathbf R)
\aver{\delta\bm\Theta}{},
\label{eq:tsubs}
\end{multline}
where  we employed the definition~(\ref{eq:fi}) of Fisher
information matrix for  $\mathbf J(\bm\theta_0|\mathbf R)=\mathbf
J(\bm\theta|\mathbf R)|_{\bm\theta=\bm\theta_0}$.
Due to symmetry $\mathbf J(\bm\theta_0|\mathbf R)=[\mathbf
J(\bm\theta_0|\mathbf R)]^{\T}$ holds
\begin{equation}
\delta\bm\theta^{\T}
\mathbf J(\bm\theta_0|\mathbf R)
\aver{\delta\bm\Theta}{}
=
\frac{1}{2}\left[
\delta\bm\theta^{\T}
\mathbf J(\bm\theta_0|\mathbf R)
\aver{\delta\bm\Theta}{}
\!+\!
\aver{\delta\bm\Theta}{}^{\T}
\mathbf J(\bm\theta_0|\mathbf R)
\delta\bm\theta
\right],
\end{equation}
and so from Eq.~(\ref{eq:mutinfcm}) we have
\begin{equation}
I(\bm\Theta;\mathbf R)
\approx
\frac{1}{2}
\aver{\left[\delta\bm\theta-\aver{\delta\bm\Theta}{}\right]^{\T}
\mathbf J(\bm\theta_0|\mathbf R)
\left[\delta\bm\theta-\aver{\delta\bm\Theta}{}\right]}
{\delta\bm\theta}.
\label{eq:mi-s1}
\end{equation}
The covariance matrix $\mathbf C_{\delta\bm\Theta}$
of r.v. $\delta\bm\Theta$ is defined as
\begin{equation}
\mathbf C_{\delta\bm\Theta}=
\aver{\left[\delta\bm\theta-\aver{\delta\bm\Theta}{}\right]
\left[\delta\bm\theta-\aver{\delta\bm\Theta}{}\right]^{\T}}
{\delta\bm\theta},
\end{equation}
and obviously $\mathbf C_{\delta\bm\Theta}=\mathbf
C_{\delta\bm\Theta}^{\T}$. Since $\bm\theta_0$ is fixed,
and $\bm\Theta=\delta\bm\Theta+\bm\theta_0$, the
covariance matrices of r.v. $\bm\Theta$ and r.v. $\delta\bm\Theta$ are
equal, $\mathbf C_{\bm\Theta}=\mathbf C_{\delta\bm\Theta}$.
 Furthermore,
the law of matrix multiplication gives
$[\mathbf A\mathbf B]_{ik}=\sum_{j}[\mathbf A]_{ij}[\mathbf B]_{jk}$,
thus summing along $i=k$ gives the trace, i.e., 
$\tr(\mathbf A\mathbf B)= \sum_i[\mathbf A\mathbf B]_{ii}=
\sum_{i,j}[\mathbf A]_{ij}[\mathbf B]_{ji}$. Therefore,
Eq.~(\ref{eq:mi-s1}) can be written in a compact form as
Eq.~(\ref{eq:weakmi}).


%

\end{document}